\documentclass[aps,pra,twocolumn,groupedaddress]{revtex4-1}

\usepackage{amsmath, amsthm, amssymb, epsfig}

\newcommand{\e}{\varepsilon}

 \newtheorem{lem}{Lemma}
\newtheorem{definition}{Definition} \newtheorem{cor}{Corollary}

\begin{document}

\title{Bounds for nonlocality distillation protocols}

\author{Manuel Forster} \affiliation{Computer Science Department, ETH
  Z\"urich, CH-8092 Z\"urich, Switzerland}

\date{\today}

\begin{abstract}
  Nonlocality can be quantified by the violation of a Bell
  inequality. Since this violation may be amplified by local operations
  an alternative measure has been proposed -- distillable
  nonlocality. The alternative measure is difficult to calculate exactly
  due to the double exponential growth of the parameter space. In this
  article we give a way to bound the distillable nonlocality of a
  resource by the solutions to a related optimization problem. Our upper
  bounds are exponentially easier to compute than the exact value and
  are shown to be meaningful in general and tight in some cases.
\end{abstract}

\maketitle

\section{Introduction}

When two separated parts of a quantum state are measured in different
bases, then the outcomes can be correlated in a way that cannot be
explained by information shared before the
separation~\cite{Bell-1964}. This property has been termed {\em quantum
  nonlocality}. Nonlocal correlations have since proved to be a useful
resource leading to new applications such as, for example, quantum key
distribution~\cite{Hardy05,PhysRevLett.98.230501,PhysRevLett.102.140501,
  hanggi-2010-2010}. Stronger correlations that are still in accordance
with the nonsignaling postulate of relativity~\cite{PR-1994} can be
defined and are formalized and studied in so-called generalized
nonsignaling theories~\cite{barrett05,barret-2005} in which quantum
correlations are a special case. Assuming nonlocality that is
super-quantum to some extent has interesting consequences for nonlocal
computation~\cite{Linden-2006} and for communication
complexity~\cite{brassard-2006,brunner2009}. {\em Maximal\/} nonlocality
allows to compute every distributed Boolean function with just one
communicated bit~\cite{vandam-2005}.

Local correlations obey certain linear constraints, so-called Bell
inequalities~\cite{Bell-1964}. The extent by which a Bell inequality is
violated by a correlation can be taken as a measure for nonlocality
(see~\cite{PhysRevLett.101.050403,Elitzur199225,fitzi2010} for other
nonlocality measures and related results). The question arose whether
this measure is also meaningful for quantifying the nonlocality of a
resource consisting of several copies of a correlation. In a context
where nonlocality is more useful the stronger it is this measure is
problematic if stronger nonlocality can be obtained from a number of
weakly nonlocal correlations. In other words, we want to know if two
parties having access to correlations violating a Bell inequality by
some small extent can execute local operations to obtain a higher
violation. The two parties can carry out arbitrary local operations, but
{\em cannot\/} communicate. This process is called {\em nonlocality
  distillation}.

It has recently been shown that the Clauser-Horner-Shimony-Holt
inequality~\cite{CHSH-1969} (CHSH) as a measure for nonlocality in
minimal dimensions is indeed problematic because it is distillable in
general~\cite{fww-2009}. A protocol by Brunner and
Skrzypczyk~\cite{brunner2009} manages to distill an arbitrarily weak
nonlocal correlation, that is still super-quantum, to the extent where
communication complexity collapses. As proposed in \cite{fww-2009} a
more meaningful measure, which is by definition undistillable, is the
maximal CHSH violation achievable from many realizations of a given
nonlocal resource by any distillation protocol. Brunner et
al.~\cite{PhysRevLett.106.020402} recently compared this alternative
measure, termed {\em distillable nonlocality}, to the
Elitzur-Popescu-Rohrlich (EPR2) decomposition
approach~\cite{Elitzur199225,fitzi2010}. The authors discovered examples
of bound nonlocality and activation -- the box-world analogues to bound
entanglement and its activation.

H{\o}yer and Rashid~\cite{Hoyer2010} thoroughly analyzed different
distillation protocols, proved optimality in restricted classes and
discovered a strong dependence between the optimal protocol and the
parameters of the resource that it distills. Today our knowledge about
general distillation protocols is still very limited. Although tools and
techniques to analyze fixed protocols exist (for example by discrete
maps~\cite{brunner2009,PhysRevA.80.062107}), difficulties arise and the
lack of a neat framework becomes obvious when attempting to calculate
the distillable nonlocality of a fixed resource. Besides what results
from some trivial symmetries not much is known to simplify the task. It
usually requires an exhaustive search over all possible distillation
protocols. The fact that the search space grows doubly exponential in
the number of involved correlation copies represents a serious problem
for the usefulness and understanding of this alternative nonlocality
measure. Therefore, as recently suggested in
\cite{PhysRevLett.106.020402}, meaningful bounds for distillable
nonlocality that can be derived more efficiently are a reasonable
compromise.

In this article we propose such bounds and analyze them. They rely on a
recursive optimization problem that has a strong connection to the
process of optimizing distillation protocols for isotropic
nonlocality. Solved with a dynamic-programming approach our bounds yield
an exponential gain in the run-time compared to a brute-force
search. The efficiency advantage enables one to analyze instances with
the number of available correlations limited by 9. To demonstrate this
we calculated the distillable nonlocality bounds for a set of fixed
isotropic systems and found that, in this context, our bounds rule out
distillation and are therefore tight. Furthermore, we present a general
idea how to extend isotropic bounds to the non-isotropic case. Here, our
calculated bounds show that the distillable nonlocality approaches the
CHSH nonlocality of a single copy when the resource is chosen closer to
an isotropic line.

\section{Preliminaries}

Following a general approach to nonlocality~\cite{barret-2005}, we
consider correlations in the joint behavior of the two ends of a
bipartite input-output {\em system}, characterized by joint probability
distributions $P^{xy}_{AB}$ on random variables $A,B$ for each
$x,y$. Let $x$ and $a$ be the input and output on the left-hand side of
the system, and $y$ and $b$ the corresponding values on the right-hand
side. On inputs $x$ and $y$ the system returns outputs $a$ and $b$ with
probability $P^{xy}_{AB}(a,b)$.

\begin{figure}[h]
  \includegraphics{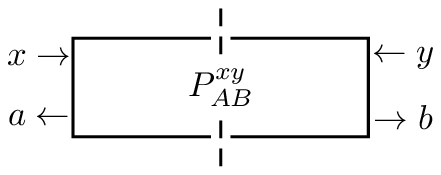}
  \label{fig1}
\end{figure}

In a distillation setting two parties, Alice and Bob, sharing a number
of systems can independently choose inputs -- possibly derived from
local outputs of other systems -- and collect outputs on their ends of
the systems. In other words, they can apply any classical circuitry to
their local parts of the systems. Such a local input-output strategy is
called a {\em wiring}~ \cite{barret-2005,wiring2006}. A party receives
its output from a system immediately after giving its input,
independently of whether the other has given its input already. This
prevents the parties from signaling by delaying their inputs. On the
other hand it allows for wirings that have a different temporal order on
the two sides. For example: While Alice's input into system 2 depends on
the output of system 1, Bob's input into system 1 depends on the output
of system 2. If the temporal order on both sides is equal we call the
wiring {\em ordered} and {\em disordered} otherwise. In this paper we
consider both kinds, summarized under general wirings as defined
in~\cite{pironio-2005}. Since we are only interested in optimality
ignoring non-deterministic strategies is sufficient due to
convexity. Thus, we are allowed to ignore randomness shared between the
two players to simplify the proofs.

Here, we assume nonlocality in its simplest form, namely, in the {\em
  binary} setting where each input and each output has two possible
values, i.e., $a,b\in\{0,1\}$ and $x,y\in\{0,1\}$. We refer to
\cite{barret-2005} for a detailed description of the convex polytope of
binary nonsignaling systems ($\mathcal{NS}$). In the case where both
inputs and both outputs are binary, the only Bell inequality (up to
symmetries) is the CHSH inequality. Furthermore, the set of eight CHSH
inequalities is complete for binary systems in the sense that if none of
them is violated, then the system is local. A system is called {\em
  isotropic} if it is a probabilistic mixture of opposite nonlocal
vertices of $\mathcal{NS}$, which are known as
PR-boxes~\cite{PR-1994}. See \cite{dw-2008} for a formal definition of
isotropic systems.

Let $P$ abbreviate any binary system. Suppose that Alice and Bob share
$n$ copies of $P$, which we index as $P_1,...,P_n$. The outputs of the
$i$-th system $P_i$ are denoted $(a_i,b_i)\in\{0,1\}^2$. So, the binary
strings $a=(a_1,a_2,...,a_n)\in\{0,1\}^n$ and
$b=(b_1,b_2,...,b_n)\in\{0,1\}^n$ are the outputs of the $n$ systems to
Alice and to Bob. Also, from now on, let $A=(A_1,...,A_n)$ and
$B=(B_1,...,B_n)$ denote random variables for the strings of the $n$
collected bits on both sides. Any wiring of $n$ shared copies of a
system then fully determines a joint distribution $W_{AB}$ on the space
of all sequences $(a,b)\in\{0,1\}^{n\times n}$. Alice and Bob can
condition their wirings on the inputs $x$ and $y$. We denote a system of
wiring distributions on $n$ copies of the binary system $P$ by $P^n$. We
will use the matrix representation
\[
P^n=\left[\begin{array}{ll}W_{AB}^{00}&W_{AB}^{01}
    \\
    W_{AB}^{10}&W_{AB}^{11}\end{array}\right]
\]
where $W_{AB}^{xy}$, for any $x,y\in\{0,1\}$, is a $2^n\times 2^n$
matrix with the probability $W_{AB}^{xy}(a,b)$ at position $(a,b)$. As
the final step in a distillation attempt the two parties map their
inputs $x$ and $y$ and the collected strings $a$ and $b$ to one local
bit each. For all $x,y$ let $f_x,g_y:\{0,1\}^n\rightarrow\{0,1\}$ stand
for these local Boolean functions. It will often be convenient to write
them as truth tables $f_x,g_y\in\{0,1\}^{2^n}$. Now we have all the
ingredients to define a distillation protocol on binary systems.

\begin{definition}
  A deterministic nonlocality distillation protocol on $n$ copies of a
  system $P$, denoted $DP=(P^n,f_0,f_1,g_0,g_1)$, consists of a system
  of wiring distributions $P^n$ and local decision functions $f_0,f_1$
  and $g_0,g_1$ for both inputs and both sides.
\end{definition}

We refer to \cite{CHSH-1969} for an exact formulation of the CHSH
nonlocality and to \cite{fww-2009} for the definition of the CHSH
nonlocality of a single binary system $NL(P)$ as required here. Note
that $NL(P)>2$ indicates that $P$ can be used to violate a CHSH
inequality and is therefore called nonlocal. To measure the violation of
a CHSH inequality by a distillation protocol we use a slightly
different representation. Observe that for each input pair $x,y$ the
correlation function of the binary system simulated by a distillation
protocol can be described by the following inner product:
\begin{align*}
  \langle
  f_x,g_y\rangle&=\mathbb{E}[f_x(a)=g_y(b)]-\mathbb{E}[f_x(a)\neq
  g_y(b)]
  \\
  &=\sum_{ab}W_{AB}^{xy}(a,b)(1-2f_x(a))(1-2g_y(b))
  \\
  &=(1-2f_x)^TW_{AB}^{xy}(1-2g_y).
\end{align*}
The two expectation values are calculated over all pairs $(a,b)\sim
W_{AB}^{xy}$. The nonlocality of a protocol is the maximal violation of
a CHSH inequality by the simulated system.
\begin{definition}\label{d6}
  The CHSH value of a distillation protocol $DP=(P^n,f_0,f_1,g_0,g_1)$
  is measured by
  \[
  NL(DP)=\max_{xy}\left|\langle f_x,g_y\rangle+\langle
    f_{\bar{x}},g_y\rangle+\langle f_x,g_{\bar{y}}\rangle-\langle
    f_{\bar{x}},g_{\bar{y}}\rangle\right|
  \]
  where we use $\bar{x}$ and $\bar{y}$ to indicate bit flips, that is,
  $\bar{0}=1$ and $\bar{1}=0$. Note that $NL(DP)>2$ indicates that the
  protocol simulates a nonlocal system, whereas $NL(DP)>NL(P)$ indicates
  a successful distillation attempt.
\end{definition}

We will henceforth assume that the maximum of this expression is reached
at $xy=00$. Our reasoning can easily be extended to all four cases.

Motivated by the discovery of successful distillation protocols for CHSH
nonlocality another measure has been proposed that expresses the
nonlocality of the optimal protocol on a given resource. Brunner et
al.~\cite{PhysRevLett.106.020402} have recently formally defined and
analyzed distillable nonlocality in comparison with the EPR2
decomposition~\cite{Elitzur199225,fitzi2010}. Here, we use a slightly
different notation that expresses the same idea.
\begin{definition}
  The distillable nonlocality of $n$ copies of a system $P$ is defined
  as
  \begin{align}
    D(n,P)=\max_{DP}NL(DP),\label{dist nonl}
  \end{align}
  where we maximize over all deterministic nonlocality distillation
  protocols on $n$ copies of $P$.
\end{definition}
For determining distillable nonlocality precisely by an exhaustive search
over the space of all possible sequences $(P^n,f_0,f_1,g_0,g_1)$ one
requires to test roughly
$\left(\prod_{i=0}^n2^{2^{i+1}}\right)^2=2^{\sum_{i=0}^n2^{i+2}}\in
2^{2^{O(n)}}$ instances, calculated as the product of all involved
Boolean functions. Because this is infeasible with normal hardware
already for $n>2$, giving an asymptotic bound to $D(n,P)$ instead which
can be computed more efficiently, is a reasonable compromise.

\section{Results}

First, we show an upper bound on $D(n,P)$ if $P$ is an isotropic system
(Corollary \ref{t1}) and move on to a general upper bound afterward (Corollary
\ref{t2}).  We start by manipulating (\ref{dist nonl}) to reveal the
core of the task. One can group the Boolean functions $f_0,f_1,g_0,g_1$,
over which we optimize in (\ref{dist nonl}), by their output
distribution. For a given number $n$, a class $C_{k}$ of Boolean
functions is defined as
\[
C_{k}=\{f:\{0,1\}^n\rightarrow\{0,1\}:|f^{-1}(1)|=k\}.
\]
We can now rewrite the distillable nonlocality as
\[
D(n,P)=\max_{0\leq k^x,l^y\leq 2^n}\max_{DP}NL(DP),
\]
where for each $k^x,l^y$ we maximize over all distillation protocols
restricted by $f_x\in C_{k^x}$ and $g_y\in C_{l^y}$ for all $x,y$. Now
concentrate on $\max_{DP}NL(DP)$ with fixed $k^x$ and $l^y$ for all
$x,y$. According to Definition \ref{d6}, this subproblem boils down to
finding the system of wiring distributions $P^n$ and functions $f_0\in
C_{k^0},f_1\in C_{k^1},g_0\in C_{l^0},g_1\in C_{l^1}$, such that the
term
\[
\langle f_0,g_0\rangle+\langle f_1,g_0\rangle+\langle
f_0,g_1\rangle-\langle f_1,g_1\rangle
\]
is minimal/maximal. Since for any $x,y$ we have
\begin{align*}
  \langle
  f_x,g_y\rangle&=(1-2f_x)^TW_{AB}^{xy}(1-2g_y)\\
  &=1-k^x/2^{n-1}-l^y/2^{n-1}+4f_x^TW_{AB}^{xy}g_y,
\end{align*}
bounds for the correlation functions $\langle f_x,g_y\rangle$ can be
derived from optimizing $f^TW_{AB}g$, with $W_{AB}$ a wiring
distribution and $f\in C_{k^x}$, $g\in C_{l^y}$, independently for all
$x,y$. We will derive bounds to $f^TW_{AB}g$ by slightly relaxing the
constraint that $W_{AB}$ needs to be a distribution obtained from a
wiring. This relaxation has the following background:

Suppose fixed functions $f\in C_{k^x},g\in C_{l^y}$ and a fixed wiring
distribution $W_{AB}$. Let $(a_1,b_1)$ be the outputs of the first
system Alice accesses in this wiring. The $k^x$ preimages of $1$ under
$f$ are split into two parts with reference to the output bit $a_1$ and
the $l^y$ preimages of $1$ under $g$ are split into two parts according
to $b_1$, i.e., we have
\[
k^x=|\{f^{-1}(1):a_1=0\}|+|\{f^{-1}(1):a_1=1\}|
\]
and
\[
l^y=|\{g^{-1}(1):b_1=0\}|+|\{g^{-1}(1):b_1=1\}|.
\]
The wiring $W_{AB}$ provides a probability distribution on $(a_1,b_1)$,
possibly conditioned on outputs of some other shared systems. Roughly
speaking we will find the optimal splitting of $f,g$ and the optimal
distribution for the output pair $(a_1,b_1)$ by considering the optimal
splittings and optimal distributions in the four subproblems, where the
pair $(a_1,b_1)$ is assumed to have fixed values. In this way we obtain
a recursive $n$-level hierarchy of optimization problems. The following
definition makes this idea precise.

A problem instance is given by the tuple $(P,n,k,l)$, with a binary
system $P$, integers $n>1$ and $0 \leq k,l\leq 2^n$. The recursive
optimization to solve is:

\begin{definition}\label{problem}
  Given any instance $(P,n,k,l)$ fix the parameter $p=P_{AB}^{00}(0,0)$
  and find the values
  \[
  \delta_n^+(k,l)=
  \begin{cases}
    \max_{ij}
    &p[\delta_{n-1}^+(i,j)+\delta_{n-1}^+(k-i,l-j)]\\
    &\hspace{-0.8cm}+(\frac{1}{2}-p)[\delta_{n-1}^+(i,l-j)+\delta_{n-1}^+(k-i,j)],\\
    &\text{if }n>0,\\
    kl,&\text{otherwise,}
  \end{cases}
  \]
  and
  \[
  \delta_n^-(k,l)=
  \begin{cases}
    \min_{ij}&p[\delta_{n-1}^-(i,j)+\delta_{n-1}^-(k-i,l-j)]\\
    &\hspace{-0.8cm}+(\frac{1}{2}-p)[\delta_{n-1}^-(i,l-j)+\delta_{n-1}^-(k-i,j)],\\
    &\text{if }n>0,\\
    kl,&\text{otherwise,}
  \end{cases}
  \]
  where we optimize over the ranges $k-\min(k,2^{n-1})\leq
  i\leq\min(k,2^{n-1})$ and $l-\min(l,2^{n-1})\leq j\leq
  \min(l,2^{n-1})$.
\end{definition}

Solutions $\delta_n^+(k^x,l^y)$ and $\delta_n^-(k^x,l^y)$ to the
instance $(P,n,k^x,l^y)$ capture the possibilities we have when
constructing a wiring distribution $W_{AB}$ and functions $f\in
C_{k^x}$, $g\in C_{l^y}$, such that the product $f^TW_{AB}g$ is
optimal. Since in the above optimization the subproblems on the same
level are handled independently we get better solutions than with
distributions from real wirings, where this independence is not given.

\begin{lem}\label{mainLemma}
  Let $P$ be any isotropic system and suppose integers $0\leq k,l\leq
  2^n$. For any wiring distribution $W_{AB}$ on $n$ copies of $P$ and
  functions $f\in C_{k},g\in C_{l}$ we have
  \begin{align*}
    \delta_n^-(k,l)\leq f^TW_{AB}g\leq\delta_n^+(k,l),
  \end{align*}
  where $\delta_n^-(k,l)$ and $\delta_n^+(k,l)$ are the solutions to
  $(P,n,k,l)$.
\end{lem}

\begin{proof}
  The idea for proving the upper bound is to decompose $W_{AB}$ into
  four subwirings and $f,g$ into two functions each, on which we then
  inductively prove the statement. The decomposition is done with
  reference to the first system Alice accesses, called $P_1$, its
  outputs denoted $(a_1,b_1)$. We write $a_{\bar{1}}$ for the string
  $(a_2,...,a_n)$ and $b_{\bar{1}}$ for $(b_2,...,b_n)$ and accordingly
  we use the random variables $A_{\bar{1}}=(A_2,...,A_n)$ and
  $B_{\bar{1}}=(B_2,...,B_n)$.

  We fix the probability $p=P_{AB}^{00}(0,0)$. Let the Boolean function
  $h:\{0,1\}^{n-1}\rightarrow \{0,1\}$ determine a bit depending on Bobs
  local outputs $b_{\bar{1}}=(b_2,...,b_n)$ of the $n-1$ systems
  $P_2,...,P_n$ shared with Alice, such that
  \[
  W_{A_1B_1|B_{\bar{1}}}(0,h(b_{\bar{1}})|b_{\bar{1}})=p
  \]
  for all $b_{\bar{1}}\in\{0,1\}^{n-1}$. Note that if the wiring
  defining $W_{AB}$ is ordered then Alice's and Bobs inputs into the
  first system are constant and thus $h$ is also constant. Otherwise,
  only Alice's input is guaranteed to be constant. In this case $h$ is
  used to react to possible changes caused by Bobs input into the first
  system. $h$ always exists since isotropic systems have symmetric
  output distributions, and therefore for any $b_{\bar{1}}$ the first
  system $P_1$ outputs either $(0,0)$ or $(0,1)$ with probability $p$.

  According to the result of $h$ we split $W_{AB}$ into four parts by
  grouping all probabilities $W_{AB}(a,b)$ by four possible outcome
  pairs of the first system. In the same manner we split the functions
  $f,g$ into the subfunctions $f'(a_{\bar{1}})=f(0,a_{\bar{1}})$ and
  $f''(a_{\bar{1}})=f(1,a_{\bar{1}})$ for any $a_{\bar{1}}$ and
  $g'(b_{\bar{1}})=g(h(b_{\bar{1}}),b_{\bar{1}})$ and
  $g''(b_{\bar{1}})=g(\bar{h}(b_{\bar{1}}),b_{\bar{1}})$ for any
  $b_{\bar{1}}$. Applied to $f'W_{AB}g$ we can identify four parts as
  \begin{align}
    f^TW_{AB}g&=\sum_{ab}f(a)W_{AB}(a,b)g(b)\notag
    \\
    &=\sum_{a_{\bar{1}}b_{\bar{1}}}\left[f'(a_{\bar{1}})W_{AB}(0a_{\bar{1}},h(b_{\bar{1}})b_{\bar{1}})g'(b_{\bar{1}})\right.\notag
    \\
    &~~~~~~+f''(a_{\bar{1}})W_{AB}(1a_{\bar{1}},\bar{h}(b_{\bar{1}})b_{\bar{1}})g''(b_{\bar{1}})
    \notag\\
    &~~~~~~+f''(a_{\bar{1}})W_{AB}(1a_{\bar{1}},h(b_{\bar{1}})b_{\bar{1}})g'(b_{\bar{1}})\notag
    \\
    &~~~~~~+\left.f'(a_{\bar{1}})W_{AB}(0a_{\bar{1}},\bar{h}(b_{\bar{1}})b_{\bar{1}})g''(b_{\bar{1}})\right].\notag
  \end{align}
  Now observe the probability
  $W_{AB}(0a_{\bar{1}},h(b_{\bar{1}})b_{\bar{1}})$. By the definition of
  the function $h$ we can derive
  \begin{align*}
    W_{A_1B_1}(0,h(b_{\bar{1}}))=\sum_{b_{\bar{1}}}W_{B_{\bar{1}}}(b_{\bar{1}})
    W_{A_1B_1|B_{\bar{1}}}(0,h(b_{\bar{1}})|b_{\bar{1}})=p
  \end{align*}
  and therefore, if we factor out the distribution of the first output
  pair $(a_1,b_1)$ we get
  \begin{align*}
    W_{AB}(0a_{\bar{1}},h(b_{\bar{1}})b_{\bar{1}})=pW_{A_{\bar{1}}B_{\bar{1}}|A_1B_1}(a_{\bar{1}},b_{\bar{1}}|0,h(b_{\bar{1}})).
  \end{align*}
  For any $a_{\bar{1}},b_{\bar{1}}$ the probability
  $W_{A_{\bar{1}}B_{\bar{1}}|A_1B_1}(a_{\bar{1}},b_{\bar{1}}|0,h(b_{\bar{1}}))$
  is independent of the distribution of the first output pair. That
  means $P_1$ can as well be replaced by local circuitry. An alternative
  wiring distribution is obtained from the original wiring by fixing
  $a_1=0$ and $b_1=h(b_{\bar{1}})$ for any
  $a_{\bar{1}},b_{\bar{1}}$. Therefore, there exists a wiring
  distribution $W_{A_{\bar{1}}B_{\bar{1}}}$ on $n-1$ systems, for
  example the one described above, such that for each
  $a_{\bar{1}},b_{\bar{1}}$ we have
  \[
  W_{AB}(0a_{\bar{1}},h(b_{\bar{1}})b_{\bar{1}})=pW_{A_{\bar{1}}B_{\bar{1}}}(a_{\bar{1}},b_{\bar{1}}).
  \]
  The above reasoning can be applied to all four parts of $f^TW_{AB}g$
  implying the existence of wiring distributions
  $W_{A_{\bar{1}}B_{\bar{1}}}^{(1)},W_{A_{\bar{1}}B_{\bar{1}}}^{(2)},W_{A_{\bar{1}}B_{\bar{1}}}^{(3)}$
  and $W_{A_{\bar{1}}B_{\bar{1}}}^{(4)}$, such that
  \begin{align*}
    f^TW_{AB}g&=p\left(f'^TW_{A_{\bar{1}}B_{\bar{1}}}^{(1)}g'+f''^TW_{A_{\bar{1}}B_{\bar{1}}}^{(2)}g''\right)\\
    &+(1/2-p)\left(f''^TW_{A_{\bar{1}}B_{\bar{1}}}^{(3)}g'+f'^TW_{A_{\bar{1}}B_{\bar{1}}}^{(4)}g''\right).
  \end{align*}
  Assume now the induction hypothesis that for any wiring distribution
  $W_{A_{\bar{1}}B_{\bar{1}}}$ on $n-1$ isotropic systems and functions
  $f,g:\{0,1\}^{n-1}\rightarrow\{0,1\}$ we have the upper bound
  \[
  f^TW_{A_{\bar{1}}B_{\bar{1}}}g\leq\delta_{n-1}^+(|f^{-1}(1)|,|g^{-1}(1)|).
  \]
  Since the basis $f^TP_{AB}g\leq \delta^+_1(|f^{-1}(1)|,|g^{-1}(1)|)$
  trivially holds, we can conclude
  \begin{align*}
    f^TW_{AB}g&\leq p\delta_{n-1}^+(|f'^{-1}(1)|,|g'^{-1}(1)|)\\
    &~~~+p\delta_{n-1}^+(|f''^{-1}(1)|,|g''^{-1}(1)|)\\
    &~~~+(1/2-p)\delta_{n-1}^+(|f'^{-1}(1)|,|g''^{-1}(1)|)\\
    &~~~
    +(1/2-p)\delta_{n-1}^+(|f''^{-1}(1)|,|g'^{-1}(1)|)\\
    &\leq\delta_n^+(k,l).
  \end{align*}
  We used that $0\leq |f'^{-1}(1)|,|f''^{-1}(1)|\leq \min(k,2^{n-1})$
  and $0\leq |g'^{-1}(1)|,|g''^{-1}(1)|\leq \min(l,2^{n-1})$ holds. With
  the same arguments one can prove the lower bound
  $f^TW_{AB}g\geq\delta_n^-(k,l)$ by induction.\end{proof}

The shown bounds to $f^TW_{AB}g$ yield the following bound to the
distillable nonlocality of isotropic systems:
\begin{cor}\label{t1}
  Let $P_{iso}$ be an isotropic system; then
  \begin{align*}
    D(n,P_{iso})\leq&\max_{0\leq k^x,l^y\leq
      2^n} 2-\frac{(k^0+l^0)}{2^{n-2}}+4[\delta_n^+(k^0,l^0)\\
    & +\delta_n^+(k^0,l^1)+\delta_n^+(k^1,l^0)-\delta_n^-(k^1,l^1)].
  \end{align*}
  where $\delta_n^+(k^x,l^y)$ and $\delta_n^-(k^x,l^y)$ are solutions to
  the problem instance $(P_{iso},n,k^x,l^y)$ for all $x,y$.
\end{cor}
\begin{proof} Note that from the definition of the inner product it
  follows that it is sufficient to consider positive values to find the
  maximum (see Lemma \ref{al3} in Appendix \ref{a3}). Therefore we can
  omit the modulus in the CHSH value and get
  \begin{align*}
    D(n,P)&=\max_{0\leq k^x,l^y\leq 2^n}\max_{DP}~~~NL(DP)\\
    &=\max_{0\leq k^x,l^y\leq 2^n}\max_{DP}~~~\langle
    f_0,g_0\rangle+\langle f_1,g_0\rangle\\
    &~~~~~~~~~~~~~~~~~~~~~~~~~~ +\langle f_0,g_1\rangle-\langle
    f_1,g_1\rangle
  \end{align*}
  where the second maximum is over $f_x\in C_{k^x}$ and $g_y\in C_{l^y}$
  for all $x,y$. To complete the proof use $\langle
  f_x,g_y\rangle=1-k^x/2^{n-1}-l^y/2^{n-1}+4f_x^TW_{AB}^{xy}g_y$ for all
  $x,y$ and apply Lemma \ref{mainLemma}. \end{proof}

From the bound on distillable nonlocality of isotropic resources we
derive a general bound as follows.

\begin{cor}\label{t2}
  For any $P$, we have $D(n,P)\leq D(n,P_{iso})$ where $P_{iso}$ shall
  be isotropic and minimize $NL(P_{iso})$ under the constraint
  $P=qP_{iso}+P_{L}(1-q)$ for any $q\in[0,1]$ and any local system
  $P_L$.
\end{cor}
\begin{proof}
  By contradiction. Assume $n$ $P$'s can be distilled to a value higher
  than $D(n,P_{iso})$. Given $n$ $P_{iso}$'s we can construct a
  distillation protocol as follows: Combine each of the $n$ $P_{iso}$'s
  with one local system $P_L$ to get the $n$ mixtures
  $P=qP_{iso}+(1-q)P_L$. Since we assumed that $n$ $P$'s can be
  distilled above $D(n,P_{iso})$ we constructed a protocol on $n$
  $P_{iso}$'s that achieves the same.
\end{proof}
Bounding $D(n,P_{iso})$ by Corollary \ref{t1} yields an exponential
speedup compared to the naive search through all protocols. The idea is
to use dynamic programming to solve the problem in Definition
\ref{problem} in run-time $2^{O(n)}$ (Lemma \ref{al1} in Appendix
\ref{a1}). Bounding $D(n,P)$ by Corollary \ref{t2} requires finding the
least nonlocal isotropic system $P_{iso}$, such that $P$ can be obtained
from a convex combination of $P_{iso}$ and a local system $P_L$. This
can be done efficiently using a linear program (see Lemma \ref{al2}
Appendix \ref{a2}). Finally we calculate a bound on $D(n,P_{iso})$ as
described above to complete the bound to $D(n,P)$.

To compare the results, an exhaustive search to determine $D(n,P)$
exactly requires an effort of $2^{2^{O(n)}}$ -- an exponential increase
to our solution.

\section{Application}

Here, we illustrate Corollaries \ref{t1} and \ref{t2} on example
applications. For this purpose consider a two-dimensional section in the
binary nonsignaling polytope defined as the convex combination of a
nonlocal vertex $P_{NL}$ (a PR-box), an unbiased mixture of two local
vertices $P_C$ (perfectly correlated random bits) and the local
isotropic system $P_F$ (lying on the CHSH facet corresponding to
$P_{NL}$)~\cite{brunner2009,Hoyer2010,PhysRevLett.106.020402} . We mix
as:
\[
P_{\e,\delta}=\e P_{NL}+\delta P_C+(1-\e-\delta)P_F,
\]
where $\e\in[0,1]$, $\delta\in[0,1]$ and $\e+\delta\leq 1$ to ensure
nonsignaling. By setting $\delta=0$ we obtain a set of isotropic systems
$P_{iso}(\e)=P_{\e,0}$, where $NL(P_{iso}(\e))=2(\e+1)$. These are
special when it comes to distillation because any system can be turned
into an isotropic system while preserving nonlocality (a procedure known
as depolarization~\cite{masanes-2005}). It is a long standing conjecture
that isotropic systems cannot be distilled. Strong evidence to support
this has been provided by Dukaric and Wolf~\cite{dw-2008}, who showed
that in the quantum region infinitely many examples of asymptotically
undistillable isotropic systems must exist. Also,
Short~\cite{short-2009} gave a proof for two-copy impossibility, i.e.,
he showed that $D(2,P_{iso})=NL(P_{iso})$ holds for any isotropic
system. With the new bound at hand we can deliver additional
evidence. We found that our isotropic bound (Corollary \ref{t1}),
calculated exactly for rational instances with Maple~\cite{maplesoft:_maple},
confirms $D(n^*,P_{iso}(\e))=2(\e+1)=NL(P_{iso}(\e))$ if walking on the
isotropic line with step length $\alpha$
($\e\in\{\alpha,2\alpha,...,1-\alpha\}$) in all tested cases:
\[
\begin{array}{r|lllllll}
  \alpha~&~10^{-6}&10^{-5}&10^{-4}&10^{-3}&10^{-2}&10^{-1}&4^{-1}\\\hline
  n^*~&~3&4&5&6&7&8&9
\end{array}
\]
An immediate consequence is that in these settings our bound is tight,
the optimal protocol simply imitates $P_{iso}(\e)$. Interestingly, the
bound on $D(n,P_{iso}(\e))$ drops significantly below $NL(P_{iso}(\e))$
if we restrict the protocol to use unbalanced functions, i.e., if we
demand $|f_0^{-1}(1)|,|f_1^{-1}(1)|\neq 2^{n-1}$ for example. Figure
\ref{fig2} illustrates the isotropic bound and this effect by an example
calculation.

\begin{figure}[h]
  \begin{center}
    \includegraphics*[scale=0.34]{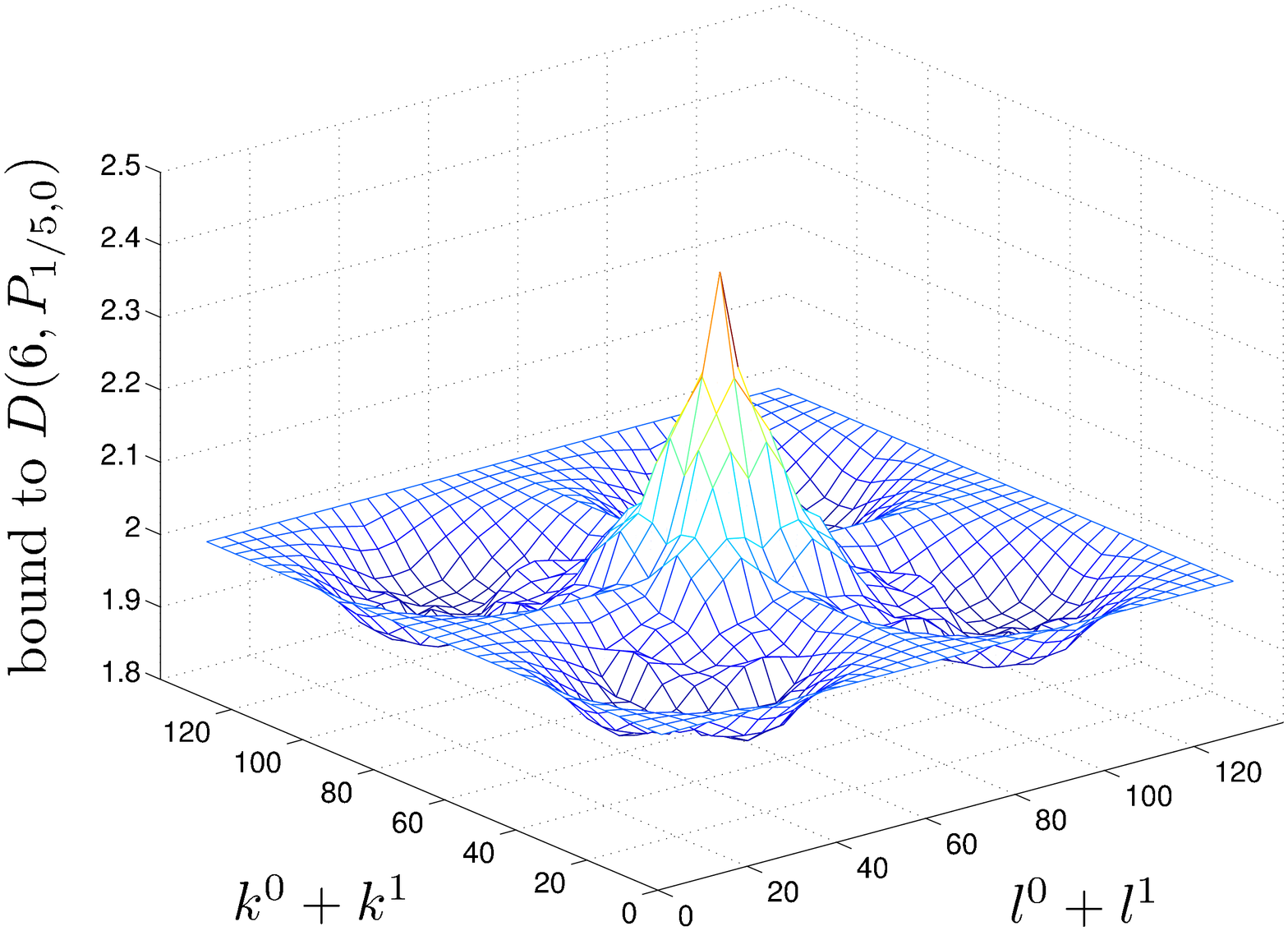}
  \end{center}
  \caption{Plotted are the calculated bounds to $D(6,P_{1/5,0})$
    ($z$-axis) for $0\leq k^0+k^1\leq 2^7$ ($x$-axis) and $0\leq
    l^0+l^1\leq 2^7$ ($y$-axis). The peak in the center is the CHSH
    value of a fully balanced protocol, i.e.,
    $|f_x^{-1}(1)|=|g_y^{-1}(1)|=2^5$ for all $x,y$. It's the only class
    that reaches $12/5$ -- exactly the value $NL(P_{1/5,0})$. Trivial
    protocols that ignore $5$ of $6$ systems achieve this value.}
  \label{fig2}
\end{figure}

Departure from the isotropic line yields another observation based on
the calculated isotropic bounds. Let $\e$ be fixed. By combining
$P_{\e,0}$ in a convex way with $P_{\e,1-\e}$ -- a system of the same
nonlocality as $P_{\e,0}$ sitting on the nonsignaling facet (also called
{\em correlated NLB}~\cite{fww-2009,brunner2009}) -- we get the system
\[
P_q=(1-q)P_{\e,0}+qP_{\e,1-\e}.
\]
The bound in Corollary \ref{t2} allows more distillation on $P_q$ the
closer the system approaches the nonsignaling facet. Starting with $q=0$
and $D(n,P_{\e,0})-NL(P_{\e,0})=0$ the difference between the bound and
$NL(P_q)$ grows with increasing weight $q$ until $D(n,P_{\e,1-\e})\leq
4$ -- an asymptotically tight bound for any $P_{\e,1-\e}$ with $\e>0$ as
confirmed by the distillation protocol in~\cite{brunner2009}.
Therefore, our tests support a suspected property of distillation
protocols, namely that, geometrically speaking, distilling a system
pushes it towards the isotropic line. The longer this distance is, the
better the system can potentially be distilled (See
Fig. \ref{fig3}). This is a property inherent to all bipartite
nonlocality distillation
protocols~\cite{fww-2009,brunner2009,Hoyer2010,PhysRevA.80.062107} known
to the author.

\begin{figure}[h]
  \includegraphics{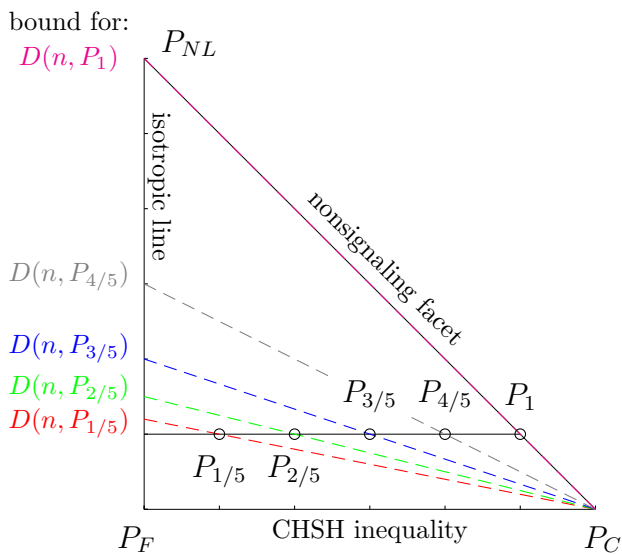}
  \caption{Here, we illustrate the idea behind the general distillation
    bound of Corollary \ref{t2} in a two-dimensional wedge. The
    distillable nonlocality of $P_q$ (for $q=1/5,2/5,3/5,4/5,1$) is
    shown to be limited by the bound to the weakest isotropic system
    $P_{iso}$, such that $P_q$ can be expressed as a convex combination
    of $P_{iso}$ and $P_C$.}
  \label{fig3}
\end{figure}

Another implication of the calculated isotropic bounds and Corollary
\ref{t2} is the existence of many sets closed under $(n<10)$-copy
wirings. We have established the following: Given the set $S(P_{iso})$
of systems that are convex combinations of a certain $P_{iso}$ and any
local system. Then $D(n,P_{iso})=NL(P_{iso})$ implies the existence of a
set $S'(P_{iso})$, where $S(P_{iso})\subseteq S'(P_{iso})$, that is
closed under $n$-copy wirings of systems in $S(P_{iso})$.

Note that the consequences of Corollary \ref{t2} are independent of how
the isotropic bounds are derived. Since the corollary expresses a
generic dependency between isotropic and general distillable nonlocality
any isotropic statement has direct implications in the general
case. Therefore, we can use known asymptotic distillation bounds to
generalize the above reasoning: Since $D(\infty,P_{iso})=NL(P_{iso})$
holds for infinitely many isotropic systems in the quantum
region~\cite{dw-2008}, the existence of infinitely many sets, each with
a different nonlocality limit, that are closed under wirings follows
directly. This answers an open question posed
in~\cite{PhysRevA.80.062107}.

\section{Conclusion}

Distillable nonlocality as a measure for the nonlocality of a given set
of correlations is sometimes preferable to the original CHSH
nonlocality, which can be amplified by local operations and is therefore
problematic in certain settings. Here, we derived bounds for the
alternative measure. Bounding distillable nonlocality is useful if the
bound can be calculated more efficiently than the exact value and is
tight enough to be meaningful. We show both properties: An increased
efficiency is achieved by a dynamic programming approach exploiting the
recursive structure of a closely related optimization problem. The bound
for isotropic systems (Corollary \ref{t1}) is tight for the
tested set of resources. It can be used to rule out distillation of up
to 9 fixed copies supporting the long standing conjecture of asymptotic
impossibility of distillation in the isotropic case. The second bound
(Corollary \ref{t2}) establishes an efficient and generic connection
between the distillable nonlocality of isotropic systems and the general
case. It explains the increasing distillation potential when departing
from the isotropic line, that has been observed in many distillation
protocols, and is tight for correlated nonlocal boxes.

The presented bounds are based on solutions to a recursive optimization
problem that is new in this context. We believe that this abstraction
can help to answer more general distillation questions and contributes
to a deeper understanding of the possibilities of such protocols. It
would be very interesting to find an efficient algorithm or an explicit
solution or bound for the defined problem. This could be an important step
towards proving general asymptotic ($n\rightarrow \infty$) bounds to
distillable nonlocality in the future.

\begin{acknowledgments}
  This work was funded by the Swiss National Science Foundation (SNSF).
\end{acknowledgments}

\appendix

\section{A symmetry of distillation protocols}\label{a3}

\begin{lem}\label{al3}
  For any deterministic distillation protocol $DP=(P^n,f_0,f_1,g_0,g_1)$
  it holds that
  \begin{align*}
    NL(P^n,f_0,f_1,g_0,g_1)&=NL(P^n,\bar{f}_0,\bar{f}_1,g_0,g_1)\\
    &
    =NL(P^n,f_0,f_1,\bar{g}_0,\bar{g}_1)\\
    & =NL(P^n,\bar{f}_0,\bar{f}_1,\bar{g}_0,\bar{g}_1)
  \end{align*}
  with $\bar{f}_x(a)=1-f_x(a)$ and $\bar{g}_y(a)=1-g_y(a)$ for all
  $a\in\{0,1\}^n$ and all $x,y\in\{0,1\}$.
\end{lem}

\begin{proof}
  Observe that for any $f,g\in\{0,1\}^{2^n}$ and any wiring distribution
  $W_{AB}$ on $n$ systems we have
  \begin{align*}
    \langle f,g\rangle&=(1-2f)^TW_{AB}(1-2g)\\
    &=-(1-2(1-f))^TW_{AB}(1-2g)\\
    &=-(1-2\bar{f})^TW_{AB}(1-2g)\\
    &=-\langle \bar{f},g\rangle.
  \end{align*}
  Apply it on
  \[
  NL(DP)=\left|\langle f_0,g_0\rangle+\langle f_0,g_1\rangle+\langle
    f_1,g_0\rangle-\langle f_1,g_1\rangle\right|
  \]
  to prove the lemma.
\end{proof}

\section{Calculate $\delta_n^+$ and $\delta_n^-$ by dynamic
  programming}\label{a1}

\begin{lem}\label{al1}
  Given numbers $n>1$ and $0\leq k,l\leq 2^n$, the values
  $\delta_n^+(k,l)$ and $\delta_n^-(k,l)$ can be calculated in
  $2^{O(n)}$ steps.
\end{lem}
\begin{proof}
  Taking advantage of the recursive structure of $\delta_n^+$ and
  $\delta_n^-$ we build the two tables
  \[
  D^+_n=\{\delta_n^+(i,j)\}_{0\leq i,j\leq 2^n}\text{ and
  }D^-_n=\{\delta_n^-(i,j)\}_{0\leq i,j\leq 2^n}
  \]
  recursively from the tables $D_{n-1}^+$ and $D_{n-1}^-$. By reading
  out from $D_{n-1}^+$ and $D_{n-1}^-$ the best of all possible
  solutions on $(i,j)$ for every entry $D_n^+(i,j)=\delta_n^+(i,j)$ and
  $D_n^-(i,j)=\delta_n^-(i,j)$ we fill the tables $D^+_n$ and $D^-_n$
  element-wise. To save time some simple reduction rules derivable from
  the inner product definition and the trivial symmetries
  $\delta_n^+(i,j)=\delta_n^+(j,i)$ and
  $\delta_n^-(i,j)=\delta_n^-(j,i)$ for each $i,j$ can be used here --
  only an eighth of the table must actually be calculated. The base
  tables are given by $D_0^+(k,l)=1$ and $D_0^-(k,l)=1$. Once the tables
  are filled, for fixed $n$ and $P$ the value $D(n,P)$ can be bounded
  with the data in $D_n^+$ and $D_n^-$ by iterating over all possible
  Boolean function classes $0\leq |f_x^{-1}(1)|\leq 2^n$ and $0\leq
  |g_y^{-1}(1)|\leq 2^n$. It requires roughly $\sum_{i=2}^n2^{4i-5}\in
  2^{O(n)}$ steps to fill the tables in this way -- a workload that is
  feasible for $n<10$ on a normal PC.
\end{proof}

\section{Finding the optimal $P_{iso}$}\label{a2}

Bounding $D(n,P)$ by Corollary \ref{t2} requires finding the least
nonlocal isotropic system $P_{iso}$ such that $P$ can be obtained from a
convex combination of $P_{iso}$ and a local system $P_L$.

\begin{lem}\label{al2}
  Given a system $P$ one can efficiently determine the isotropic system
  $P_{iso}$, such that $NL(P_{iso})$ is minimal under the constraint
  $P=qP_{iso}+P_{L}(1-q)$ for any $q\in[0,1]$ and any local system
  $P_L$.
\end{lem}

\begin{proof}
  The system $P$ identifies a particular CHSH inequality, now called
  CHSH$_P$, and therefore a family of isotropic systems denoted
  \[
  P_{iso}(\e)=\e P_{NL}+(1-\e)P_F
  \]
  with $\e\in[0,1]$ and $P_{NL}$ as the nonlocal vertex corresponding
  to CHSH$_P$ and the unique isotropic system $P_F\in \text{CHSH}_P$. To
  find the optimal $\e$ we first calculate the {\em local part of
    $P$}~\cite{Elitzur199225} with the linear program defined by Fitzi
  et al.~\cite{fitzi2010}. The local part of $P$ is the optimal value of
  the following linear program:
  \[
  \begin{array}{ll}
    \max&\sum_ip_i\\
    \text{such that}&\sum_{i}p_i{P_{L,i}}_{AB}^{xy}(a,b)\leq P_{AB}^{xy}(a,b)\\
    &p_i\geq 0.
  \end{array}
  \]
  Here, $P_{L,i}$ denote the vertices of the local polytope. Since
  $\sum_{i}p_i$ is maximal and there is only one nonlocal vertex
  violating CHSH$_P$ we have
  \[
  P=\sum_ip_iP_{L,i}+(1-\sum_ip_i)P_{NL}.\] Now, we decompose the local
  system $P^*=\sum_j\frac{p_j}{\sum_{i}p_i}P_{L,j}$ into
  \[
  P^*=p_fP_F+(1-p_f)P_L
  \]
  where $P_L$ is any local system. By the local version of Lemma 1
  in~\cite{fitzi2010} (holds since normalization and locality are linear
  properties) the maximal $p_f$ possible is given by
  \[
  p_f=\min_{a,b,x,y}\frac{{P^*}_{AB}^{xy}(a,b)}{{P_F}_{AB}^{xy}(a,b)}.
  \]
  Therefore, the parameter $\e$ fixing the optimal $P_{iso}(\e)$ can
  easily be constructed via
  \begin{align*}
    P&=\sum_ip_iP^*+(1-\sum_ip_i)P_{NL}\\
    &=\sum_ip_i(p_fP_F+(1-p_f)P_L)+(1-\sum_ip_i)P_{NL}\\
    &=\sum_ip_i(1-p_f)P_L+\sum_ip_ip_fP_F+(1-\sum_ip_i)P_{NL}\\
    &=\sum_ip_i(1-p_f)P_L+(1-\sum_ip_i(1-p_f))P_{iso}(\e),
  \end{align*}
  with
  \[
  \e=\frac{1-\sum_ip_i}{\sum_ip_ip_f+1-\sum_ip_i}.
  \]
  So, $\e$ can be determined efficiently.
\end{proof}

\end{document}